\documentclass{birkjour}

\usepackage{hyperref} 
	\hypersetup{citebordercolor=.1 .6 1}

\IfFileExists{mymtpro2.sty}{\usepackage[subscriptcorrection]{mymtpro2}
  \def\cup{\cupprod}
  \def\cap{\capprod}

  \let\widetilde\wtilde
	\def\llangle{\left\langle\mkern-4mu\left\langle}
	\def\rrangle{\right\rangle\mkern-4mu\right\rangle}
  \def\parens##1{\PARENS{\; \rule[-1cm]{0pt}{2.2cm} ##1\;}}
}{\usepackage{amssymb}
	\def\llangle{\left\langle\mkern-9mu\left\langle}
	\def\rrangle{\right\rangle\mkern-9mu\right\rangle}
	\def\parens##1{\left( ##1 \right)}
	\def\ccases##1{\begin{cases}##1\end{cases}}
}

\usepackage{upref}
\usepackage{dsfont}
\usepackage{amsmath}

\newtheorem{Theorem}{Theorem}[section]
\newtheorem{Corollary}[Theorem]{Corollary}
\newtheorem{Lemma}[Theorem]{Lemma}

\theoremstyle{definition}
\newtheorem{Def}[Theorem]{Definition}

\theoremstyle{remark}
\newtheorem{Remark}[Theorem]{Remark}
\newtheorem{myrems}[Theorem]{Remarks}
\newenvironment{Remarks}{\begin{myrems}\begin{nummer}}%
    {\end{nummer}\end{myrems}}

\newcounter{numcount}
\newcommand{\labelnummer}{(\roman{numcount})}%

\makeatletter
\providecommand{\showkeyslabelformat}[1]{\relax}        
\let\mysaveformat\showkeyslabelformat                   %
\def\myformat#1{\raisebox{-1.5ex}{\mysaveformat{#1}}}   %

\newenvironment{nummer}%
  {\let\curlabelspeicher\@currentlabel%
    \begin{list}{\textup{\labelnummer}}%
      {\usecounter{numcount}\leftmargin0pt%
        \topsep0.5ex\partopsep2ex\parsep0pt\itemsep0ex\@plus1\p@%
        \labelwidth2.5em\itemindent3.5em\labelsep1em%
      }%
    \let\saveitem\item%
    \def\item{\saveitem%
      \def\@currentlabel{\curlabelspeicher\kern.1em\labelnummer}}%
    \let\savelabel\label%
    \def\label##1{{\ifnum\thenumcount=1\let\showkeyslabelformat\myformat\fi\savelabel{##1}}%
										{\def\@currentlabel{\labelnummer}%
									 	\let\showkeyslabelformat\@gobble
									 	\savelabel{##1item}%
										}%
	   							}%
  }{\end{list}}%
  
\newenvironment{indentnummer}%
  {\let\curlabelspeicher\@currentlabel%
    \begin{list}{\textup{\labelnummer}}%
      {\usecounter{numcount}\leftmargin0pt%
        \topsep0.5ex\partopsep2ex\parsep0pt\itemsep0ex\@plus1\p@%
        \labelwidth2.5em\itemindent0em\labelsep1em%
        \leftmargin2.5em}%
    \let\saveitem\item%
    \def\item{\saveitem%
      \def\@currentlabel{\curlabelspeicher\kern.1em\labelnummer}}%
    \let\savelabel\label%
    \def\label##1{{\ifnum\thenumcount=1\let\showkeyslabelformat\myformat\fi\savelabel{##1}}%
										{\def\@currentlabel{\labelnummer}%
									 	\let\showkeyslabelformat\@gobble
									 	\savelabel{##1item}%
										}%
    							}%
  }{\end{list}}%
  
\def\itemref#1{\ref{#1item}}

\def\@cite#1#2{{%
  \m@th\upshape
  	[{#1\if@tempswa, #2\fi}]}}

\makeatother

\newcommand{\HN}{H_{\Lambda}^N}
\newcommand{\HD}{H_{\Lambda}^D}

\newcommand{{\HHM}}{$\begin{pmatrix} \HN & B\\ B& -\HD \end{pmatrix}$}

\newcommand{\Hh}{\mathds{H}}

\newcommand{\la}{\lambda}
\newcommand{\be}{\beta}
\newcommand{\nn}{\nonumber}
\newcommand{\one}{\mathbf{1}}
\def\11{\mathds{1}}
\def\bbGamma{\mathrm{I}\mkern-3.5mu\Gamma} 
\newcommand{\EE}{\mathbb E}
\newcommand{\NN}{\mathbb N}
\newcommand{\PP}{\mathbb P}
\newcommand{\RR}{\mathbb R}
\newcommand{\Zd}{\mathbb{Z}^{d}}
\renewcommand{\d}{\mathrm d}
\DeclareMathOperator{\e}{e}
\DeclareMathOperator{\supp}{supp}
\DeclareMathOperator{\dist}{dist}
\DeclareMathOperator{\tr}{tr}
\let\emptyset\varnothing

\numberwithin{equation}{section}

\begin{document}
%

\title{Localization for random block operators}

\author{Martin Gebert}
\address{Mathematisches Institut, Universit\"at M\"unchen,
	Theresienstra\ss{}e 39, 80333 M\"unchen, Germany}
\email{gebert@math.lmu.de}

\author{Peter M\"uller}
\address{Mathematisches Institut, Universit\"at M\"unchen,
	Theresienstra\ss{}e 39, 80333 M\"unchen, Germany}
\email{mueller@lmu.de} 

\thanks{Version of 9 August 2012}

\begin{abstract}
	We continue the investigations of Kirsch, Metzger and the second-named author [J.\ Stat.\ Phys.\ \textbf{143}, 1035--1054 (2011)] on spectral properties of a certain type of random block operators.
	In particular, we establish an alternative version of a Wegner estimate and an improved result on Lifschitz tails at the internal band edges. Using these ingredients and the bootstrap multi-scale analysis, 
	we also prove dynamical localization in a neighbourhood of the internal band edges.
\end{abstract}

\maketitle

%
\section{The model and its basic properties}
%

Random block operators arise in several different fields of Theoretical Physics. In this paper we are concerned with those that are relevant to mesoscopic disordered systems such as dirty superconductors. In this context, block operators are used to describe quasi-particle excitations within the self-consistent Bogoliubov-de Gennes equations. It turns out that such block operators fall in 10 different symmetry classes \cite{PhysRevB.55.1142}. As in the previous paper \cite{RandomBlockoperators}, we will consider one particular symmetry class, class $C1$, and refer to \cite{RandomBlockoperators} for further discussions and motivations.


\noindent
Given some Hilbert space $\mathcal{H}$, we write $\mathcal L(\mathcal H)$ for the Banach space of all bounded linear operators from $\mathcal H$ into itself. In this paper we are concerned with the Hilbert space $\mathcal H^{2} := \ell^2(\mathbb Z^d)\oplus \ell^2(\mathbb Z^d)$, the direct sum of two Hilbert spaces of complex-valued, square-summable sequences indexed by the $d$-dimensional integers $\Zd$.
We also fix a probability space $(\Omega,\mathcal F, \mathbb P)$ with corresponding expectation denoted by $\EE$.


\begin{Def}\label{Def:Hh}
	In this paper a \emph{random block operator} is an operator-valued random variable
	\begin{equation}
		\label{Hh-eq}
		\Hh: 
		\begin{array}{l}
 			\Omega \longrightarrow \mathcal L(\mathcal{H}^{2}) \\[1ex]
			\omega \longmapsto \Hh^\omega:= \begin{pmatrix} H^\omega & B^\omega \\ 
			                                                B^\omega &-H^\omega \end{pmatrix}
	 	\end{array}
	\end{equation}
	with the following three properties:
	\begin{indentnummer}
	\item \label{diagonal}
		For $\mathbb{P}$-a.e.\ $\omega\in\Omega$ the operator $H^\omega :=H_0+V^\omega \in 
		\mathcal{L}\big(l^{2}(\mathbb{Z}^{d})\big)$ is the discrete random Schr\"odinger operator 
		of the Anderson model. More precisely, $H_0$ stands for the negative discrete 
		Laplacian on $\mathbb{Z}^{d}$, which is defined by 
		\begin{equation}
			(H_0 \psi)(n):=-\sum_{m\in\Zd: \; |m-n|=1}\left[\psi(m)-\psi(n)\right]
		\end{equation} 
		for every	$\psi\in \ell^2(\mathbb Z^d)$ and every $n \in \mathbb{Z}^{d}$. 
		We always stick to the 1-norm $|n| := \sum_{j=1}^{d}|n_{j}|$ of 
		$n=(n_{1}, \ldots, n_{d})\in\Zd$.

		The random potential is induced by a given family 
		$(\omega \mapsto V^\omega_{n})_{n\in\mathbb Z^d}$ of i.i.d.\ real-valued random 
		variables on $\Omega$ with single-site measure $\mu_V$ of compact support in $\RR$. 
    Thus, the multiplication operator given by 
		\begin{equation}
			(V^\omega\psi)(n) := V^\omega_{n}\psi(n)
		\end{equation}
 		for every $\psi\in \ell^2(\mathbb Z^d)$ and every $n \in \mathbb{Z}^{d}$ is well-defined 
		and bounded for $\mathbb{P}$-a.e.\ $\omega\in\Omega$. Also, $H^\omega$ is self-adjoint 
		and bounded for $\mathbb{P}$-a.e.\ $\omega\in\Omega$.
	\item	\label{offdiagonal}
		For $\mathbb{P}$-a.e.\ $\omega\in\Omega$ the operator 
		$B^\omega \in \mathcal{L}\big(\ell^{2}(\mathbb{Z}^{d})\big)$ 
		is the multiplication operator induced by the family 
		$(\omega \mapsto B^\omega_{n})_{n\in\mathbb Z^d}$ 
		of i.i.d.\ real-valued random variables on $\Omega$ with single-site measure $\mu_B$ 
		of compact support in $\RR$.
	\item The family of random variables $(V_{n})_{n\in\mathbb Z^d}$ 
		is independent of the family $(B_{n})_{n\in\mathbb Z^d}$.
	\end{indentnummer}
\end{Def}

\begin{Remarks}
 	\item
	 	Conditions \itemref{diagonal} and \itemref{offdiagonal} in Definition~\ref{Def:Hh} imply that the random 
		block operator $\Hh$ is $\mathbb P$-a.s.\ self-adjoint and bounded.
	\item \label{symm-spec}
		Block	operators of the form \eqref{Hh-eq} have a spectrum that is symmetric around $0$, 
		i.e.\ $E\in\RR$ belongs to the spectrum $\sigma(\Hh^{\omega})$, if and only if this is 
		also true for $-E$ \cite[Lemma~2.3]{RandomBlockoperators}.
	\item
		The random block operator $\Hh$ is ergodic with respect to $\Zd$-translations, 
		see \cite{RandomBlockoperators} for more details. Therefore, standard results imply 
		the existence of a non-random closed set $\Sigma$ such that $\sigma(\Hh)=\Sigma$ 
		holds $\mathbb P$-a.s.\ 
		\cite{springerlink:10.1007/3-540-51783-9_23, Kirsch:Invitation, Carmona, PasturFigotin}.
		This non-randomness also extends to the components in the Lebesgue decomposition of the spectrum.
\end{Remarks}


\noindent
In order to count eigenvalues we introduce a restriction of random block operators to bounded regions of space $\Zd$.
Given $L >0$ we write $\Lambda_{L} := ]-L/2,L/2[^{d} \cap \Zd$ for the discrete cube of ``length $L$''
about the origin and $\Lambda_{L}(n):= n+ \Lambda_{L}$ for its shifted copy with centre $n\in \Zd$. 

\begin{Def}
	\label{Def:HhL}
 	Given a cube $\Lambda_{L} \subset\Zd$, we define the finite-volume Hilbert space
	$\mathcal{H}^{2}_{L} := \ell^2(\Lambda_{L})\oplus \ell^2(\Lambda_{L})$ and the 
	\emph{finite-volume random block operator}
	\begin{equation}
		\Hh_{\Lambda_{L}} \equiv \Hh_{L}: 
		\begin{array}{l}
 			\Omega \longrightarrow \mathcal L(\mathcal{H}^{2}_{L}) \\[1ex]
			\omega \longmapsto \Hh_{L}^{\omega}:= \begin{pmatrix} H_{L}^{\omega} & B^\omega \\ 
			                                                B^\omega &-H_{L}^{\omega} \end{pmatrix}
	 	\end{array},
	\end{equation}
	where $H_L:=H_{0,L}+V$ and $H_{0,L}$ is the discrete Laplacian on $\Lambda_{L}$ with simple boundary conditions. Its matrix entries are given by $H_{0,L}(n,m) := \langle\delta_n,H_0\delta_m\rangle$ for $n,m\in\Lambda_L$, with $(\delta_n)_{n\in\Zd}$ denoting the canonical basis and $\langle\boldsymbol\cdot, \boldsymbol\cdot\rangle$ the canonical scalar product of $\ell^{2}(\Zd)$ . The random multiplication operators $V$ and $B$ are restricted to $\ell^{2}(\Lambda_{L})$ in the canonical way.
\end{Def}

\begin{Remarks} 
 	\item The operator $\Hh_{L}^{\omega}$ is well-defined, bounded and self-adjoint for 
		$\PP$-a.e.\ $\omega\in\Omega$.
	\item Simple boundary conditions are sufficient for most of our purposes here. We refer to 
		\cite{RandomBlockoperators} for other useful restrictions of such types of block operators.
\end{Remarks}

\noindent
We write $|M|$ for the cardinality of a finite set $M$ and introduce the normalized 
\emph{finite-volume eigenvalue counting function}
\begin{equation}
	\mathds N_{\Hh_{L}}(E) := \frac{1}{2|\Lambda_L|} \, \big| \sigma(\Hh_L) \;\cap\; ]-\infty,E] \,\big|
                  = \frac{1}{2|\Lambda_L|} \tr_{\mathcal H_L^2}\big[ 1_{]-\infty,E]}(\Hh_L)\big],
\end{equation}
which is a non-negative random variable for every $E\in\RR$. Here, $1_{G}$ stands 
for the indicator function of a set $G$ and $\tr_{\mathcal{H}}$ for the trace over some Hilbert space 
$\mathcal{H}$. 
The existence and self-averaging of the macroscopic limit of $\mathds N_{\Hh_{L}}(E)$ is also a consequence of ergodicity. 

\begin{Lemma}[\protect{\cite[Lemma~4.8]{RandomBlockoperators}}]
 	There exists a (non-random) right-continuous probability distribution function $\mathds N:\mathbb R\rightarrow[0,1]$, the \emph{integrated density of states} of $\Hh$, and a measurable subset $\Omega_{0} \subseteq \Omega$ of full measure, $\PP(\Omega_{0})=1$, such that
\begin{equation}
\mathds N(E)=\lim\limits_{L\rightarrow\infty} \mathds N_{\Hh_{L}}^{\omega}(E)=\lim\limits_{L\rightarrow\infty} \mathbb E\left[\mathds N_{\Hh_{L}}(E)\right]
\end{equation}
holds for every $\omega\in\Omega_{0}$ and every continuity point $E\in\mathbb R$ of $\mathds N$.
\end{Lemma}


\noindent
Since $\sigma(\Hh)=\Sigma$ holds $\PP$-a.s., one can ask for the precise location of this almost-sure spectrum. A partial answer is given by 

\begin{Lemma}[\protect{\cite[Lemma~4.3]{RandomBlockoperators}}]
\label{Lemma:Spectrum1}
	Consider the random block operator $\Hh$ of Definition~\ref{Def:Hh}. Then we have $\mathbb P$-a.s.\
\begin{equation}
	\textstyle
	\Big\{\pm\sqrt{E^2+\beta^2}: E\in \sigma(H), \beta\in \supp(\mu_B)\Big\} 
	\subseteq \sigma(\Hh)\subseteq [-r,r],
\end{equation}
where $r:=\sup_{E\in \sigma(H)}|E|+ \sup_{\beta\in\supp(\mu_B)}|\beta|$.
\end{Lemma}

\noindent
We say that an interval $]a_{1}, a_{2}[$, where $a_{1}, a_{2} \in\RR$ with $a_{1} < a_{2}$, 
is a \emph{spectral gap} of a self-adjoint operator $A$, if 
$]a_{1}, a_{2}[ \;\cap \;\sigma(A) = \varnothing$ and $a_{1},a_{2} \in\sigma(A)$. 
In order to determine the spectral gap of $\Hh$, we will combine the above lemma with a deterministic result.

\begin{Lemma}[\protect{\cite[Prop.~2.10]{RandomBlockoperators}}]
\label{Lemma:Unitary}     
	Consider the random block operator $\Hh$ of Definition~\ref{Def:Hh}. Then we have for $\PP$-a.a.\ $\omega\in\Omega$:
	\begin{nummer}
	\item If there exists $\lambda \ge 0$ such that $\inf\supp \mu_{V} \ge \lambda$, then
		\begin{equation}
			\sigma (\Hh^{\omega}) \;\cap\; ]-\lambda,\lambda[ \;=\emptyset.
		\end{equation}
	\item  If there exists $\beta \ge 0$ such that $\inf\supp \mu_{B} \ge \beta$, then
		\begin{equation}
			\sigma (\Hh^{\omega})  \;\cap\;   ]-\beta,\beta[ \;=\emptyset.
		\end{equation}
	\item
		\label{doublegap} 
		If there exists $\lambda,\beta \ge  0$ such that $\inf\supp \mu_{V} \ge \lambda$ 
		and $\inf\supp \mu_{B} \ge \beta$, then
		\begin{equation}
			\textstyle
			\sigma(\Hh^{\omega})  \;\cap\;  ]-\sqrt{\lambda^2+\beta^2}, \sqrt{\lambda^2+\beta^2}[ \; =\emptyset.
		\end{equation}
	\end{nummer}
\end{Lemma}

\begin{Remark} \label{Lemma:Endpoints}
	Lemmas~\ref{Lemma:Spectrum1} and~\ref{Lemma:Unitary} together provide the following two statements.
	\begin{nummer}
	\item If $\lambda := \inf\supp \mu_{V} >0$ and $0\in \supp\mu_{B}$, 
		then $]-\lambda, \lambda[$ is $\PP$-a.s.\ a spectral gap of $\Hh$ around 0.  
	\item If $\lambda := \inf\supp \mu_{V} \ge 0$ and $\beta := \inf\supp\mu_{B} >0$, then 
	\linebreak 
		$]-\sqrt{\lambda^2+\beta^2}, \sqrt{\lambda^2+\beta^2}[$ is $\PP$-a.s.\ 
		a spectral gap of $\Hh$ around 0.
	\end{nummer}
\end{Remark}

\noindent
For completeness and later use we review the main result of \cite{RandomBlockoperators}, which is a Wegner estimate for the operator $\Hh$. In the next section we provide a new variant of this result.  
We write $\left\|f\right\|_{BV}$ for the total variation norm of some function $f:\mathbb R\rightarrow \mathbb R$.

\begin{Theorem}[Wegner estimate \protect{\cite[Thm.~5.1]{RandomBlockoperators}}]
\label{Wegner I} 
	Consider the random block operator $\Hh$ of Definition~\ref{Def:Hh} and assume that 
	at least one of the following conditions is met.
	\begin{enumerate}
	\item[\upshape{(1)}] There exists $\lambda>0$ such that $\inf\supp \mu_{V} \ge \lambda$
		 and $\mu_V$ is absolutely continuous with a piecewise continuous 
		Lebesgue density $\phi_V$ of bounded variation and compact support.
	\item[\upshape{(2)}] There exists $\beta>0$ such that $\inf\supp \mu_{B} \geq \beta$ 
		and $\mu_B$ is absolutely continuous with a piecewise continuous 
		Lebesgue density $\phi_B$ of bounded variation and compact support.
	\end{enumerate}
	Then the integrated density of states $\mathds N$ of $\Hh$ is Lipschitz continuous 
	and has a bounded Lebesgue derivative, the density of states $\mathds D:= \d\mathds N/\d E$. 
	
	Furthermore, if hypothesis \upshape{(1)} holds, then we have for Lebesgue-a.a.\ $E\in\mathbb R$ that 
	\begin{equation}
		\mathds D(E)\leq 2 \frac{|E|+1}{\lambda}\left\|\phi_V\right\|_{BV}.
	\end{equation}
	In case of hypothesis \upshape{(2)}, we get the estimate 
	\begin{equation}
		\mathds D(E)\leq 2 \frac{|E|+1}{\beta}\left\|\phi_B\right\|_{BV}
	\end{equation}
 	for Lebesgue-a.a.\ $E\in\mathbb R$.
\end{Theorem}

%
\section{Results}
%

In this section we present the results of this paper. All proofs are deferred to subsequent sections.
We start with a variant of Theorem~\ref{Wegner I}.

\begin{Theorem}[Wegner estimate]
\label{Wegner II}
	Consider the random block operator $\Hh$ of Definition~\ref{Def:Hh} and assume that 
	$\inf\supp \mu_{V} \geq0$ and $\inf\supp\mu_{B}\geq 0$. Assume further that the single-site measures $\mu_V$ and $\mu_B$ 
	are both absolutely continuous 
	with piecewise continuous Lebesgue densities $\phi_V$, $\phi_B$ of bounded variation 
	and compact support. Then the integrated density of states $\mathds N$ is Lipschitz continuous 
	with a bounded Lebesgue derivative $\mathds D = \d\mathds N/\d E$ satisfying
	\begin{equation}
		\left\|\mathds D\right\|_{\infty}\leq 2\big( 
		\left\|\phi_V\right\|_{BV}+\left\|\phi_B\right\|_{BV}\big).
	\end{equation}
\end{Theorem}

\begin{Remarks}
	\item As compared to the hypotheses of the Wegner estimate from \cite{RandomBlockoperators} 
		in Theorem~\ref{Wegner I}, the above result 
		constitutes an improvement in that neither $H$ nor $B$ have to be bounded away from $0$. The price 
    we have to pay is that \emph{both} operators are required to be non-negative 
    and that \emph{both} probability distributions are assumed to be sufficiently regular.
	\item As compared to the results of Theorem~\ref{Wegner I}, we note that the present 
		Wegner estimate is uniform in energy. 
	\item After completing this work, A.\ Elgart informed us that he can obtain a Wegner estimate for $\Hh$ which does not require assumptions on the supports of $\mu_{V}$ or $\mu_{B}$ \cite{Elg12}. 
\end{Remarks}

\noindent
Next we consider the spectral asymptotics of the integrated density of states $\mathds N$ of $\Hh$ at the internal band edges.

\begin{Theorem}[Internal Lifschitz tails -- upper bound] 
\label{LifschitzII}
	Consider the random block operator $\Hh$ of Definition~\ref{Def:Hh}. Assume that 
	$\lambda:=\inf\supp\mu_{V} \ge 0$ and that the support of the measure $\mu_V$
	consists of more than a single point.
	Assume further that one of the following conditions is met
	\begin{enumerate}
	\item[\upshape{(1)}] \quad $\beta := \inf\supp\mu_{B} \ge 0$,
	\item[\upshape{(2)}] \quad $\beta := \sup\supp\mu_{B} \le 0$,
	\item[\upshape{(3)}] \quad $0\in\supp\mu_{B}$, in which case we set $\beta:=0$.
	\end{enumerate}	  
	Then we have
	\begin{equation}
		\limsup_{\epsilon\searrow0}\frac{\ln\Big| 
		\ln\Big[ \mathds N\big(\sqrt{\lambda^2+\beta^2}+\epsilon\big) 
			- \mathds N\big(\sqrt{\lambda^2+\beta^2}\big) \Big]\Big|}{\ln\epsilon}\leq- \alpha
	\end{equation}
	with $\alpha = d/2$ in all cases except the case 
	$\lambda=0$ and $\beta \neq 0$, where $\alpha = d/4$.
\end{Theorem}

\begin{Remarks}
  \item An analogous result holds when approaching the upper edge of the lower band 
  	$ - \sqrt{\lambda^{2}+\beta^{2}}$ from below.
	\item There is no conflict in the definition of $\beta$ in Theorem~\ref{LifschitzII} 
		if several of the conditions (1) -- (3) hold, because this case is only possible with $\beta=0$.  	
  \item If $\lambda>0$ or $\beta \neq 0$, then $\pm\sqrt{\la^2+\be^2}$ are the endpoints of the 
  	spectral gap of $\Hh$; see Remark \ref{Lemma:Endpoints}. 
		To apply this remark in the case (2), use also unitary equivalence of 
		$\big(\begin{smallmatrix} H & B \\ B & -H\end{smallmatrix}\big)$ and 
		$\big(\begin{smallmatrix} H & -B \\ -B & -H\end{smallmatrix}\big)$.
	\item Theorem~\ref{LifschitzII} is a generalization of \cite[Thm.~6.1]{RandomBlockoperators}
		which applies only to $\lambda>0$ and $\beta=0$.
\end{Remarks}

\noindent
A (mostly) complementary lower bound is provided by

\begin{Theorem}[Internal Lifschitz tails -- lower bound] 
\label{Lifschitz-lower}
	Consider the random block operator $\Hh$ of Definition~\ref{Def:Hh}. Assume that 
	$\lambda:=\inf\supp\mu_{V} \ge 0$ and that one of the cases $(1)$ -- $(3)$ in 
	Theorem~\ref{LifschitzII} applies.  
 	Assume further the existence of constants $C, \kappa
	 > 0$ such that for all sufficiently small $\eta>0$ the bounds
	\begin{equation}\label{Lifschitz6}
 		\mu_{V}\big([\lambda, \lambda + \eta[\big) \ge C \eta^{\kappa} \qquad
		\text{and} \qquad
 		\mu_{B}\big(]\beta -\eta, \beta+\eta[ \big) \ge C \eta^{\kappa} 
	\end{equation}
	hold. Then we have
	\begin{equation}
		\liminf_{\epsilon\searrow0}\frac{\ln\Big| 
		\ln\Big[ \mathds N\big(\sqrt{\lambda^2+\beta^2}+\epsilon\big) 
			- \mathds N\big(\sqrt{\lambda^2+\beta^2}\big) \Big]\Big|}{\ln\epsilon} \geq - d/2. 
	\end{equation}
\end{Theorem}

\begin{Remarks}
\item 
	Taken together, Theorems~\ref{LifschitzII} and~\ref{Lifschitz-lower} imply that the 
	random block operator $\Hh$ exhibits Lifschitz tails at the edges of its spectral gap with 
	Lifschitz exponent $d/2$ for all values $\lambda >0$ and $\beta \in\RR$.
\item 
	Even in the case $\lambda = \beta =0$, the block operator $\Hh$ exhibits Lifschitz 
	tails with Lifschitz exponent $d/2$ at energy zero. We note that there is no internal spectral edge 
	at energy zero in this case.
\item
 	In the case $\lambda=0$ and $\beta \neq 0$ we believe that the correct value of the 
	Lifschitz exponent is $d/2$ (rather than $d/4$), as given by the lower bound in Theorem~\ref{Lifschitz-lower}.
\end{Remarks}

\noindent
Finally, we turn to Anderson localization of $\Hh$ in a neighbourhood of the internal band edges.
The following notion will be useful for the formulation of the result.

\begin{Def}
	\label{Def:matrix element}
	Given a bounded operator $\mathds A$ on the Hilbert space $\mathcal{H}^{2}$ and $n,m \in\Zd$, we 
	introduce its $2\times2$-matrix-valued matrix element
 	\begin{equation}
		\label{matrix-element}
 		\mathds{A}(n,m) := 
		\parens{\begin{matrix}  
					\llangle \begin{pmatrix} \delta_n \\ 0 \end{pmatrix}, 
					\mathds A
					\begin{pmatrix} \delta_m \\ 0 \end{pmatrix} \rrangle 
					&
					\llangle \begin{pmatrix} \delta_n \\ 0 \end{pmatrix}, 
					\mathds A
					\begin{pmatrix} 0 \\ \delta_m \end{pmatrix}\rrangle  
					\\[3ex] 
					\llangle \begin{pmatrix} 0 \\ \delta_n \end{pmatrix},
					\mathds A
					\begin{pmatrix} \delta_m \\ 0 \end{pmatrix}\rrangle
					&
					\llangle \begin{pmatrix} 0 \\ \delta_n \end{pmatrix},
					\mathds A
					\begin{pmatrix} 0 \\ \delta_m \end{pmatrix}\rrangle
		\end{matrix}}.
	\end{equation}
	Here $\langle\mkern-4mu\langle \boldsymbol{\cdot},\boldsymbol{\cdot}\rangle\mkern-4mu\rangle$ stands 
	for the canonical scalar product on the Hilbert space $\mathcal{H}^{2}$. We also fix some  
	norm $\|\boldsymbol\cdot\|_{2\times2}$ on the vector space of complex-valued 
	$2\times2$-matrices.
\end{Def}

\begin{Theorem}[Complete localization]
\label{Anderson Local}
	Consider the random block operator $\Hh$ of Definition~\ref{Def:Hh} and assume the hypotheses of 
	Theorem~\ref{LifschitzII}. Assume further the hypotheses of Theorem~\ref{Wegner I} or  Theorem~\ref{Wegner II}.
Then there exist constants $0< \zeta <1$, $C_{\zeta}>0$ and an energy interval $I:=[-a,a]$, where  $a>0$, such that 
$I \cap \sigma(\Hh) \neq \varnothing$ holds $\PP$-a.s.\ and 
\begin{equation}
	\label{eq:dyn-loc}
	\EE \left( \sup_{\|f\|_\infty\leq 1} \big\| \big( 1_I(\Hh)f(\Hh)\big)(n,m)\big\|_{2\times2} \right) 
	\le C_{\zeta} \e^{-|n-m|^\zeta}
\end{equation}
for all $n,m \in\Zd$. The supremum in \eqref{eq:dyn-loc} is taken over all Borel functions 
$\RR \rightarrow\mathbb C$ that are pointwise bounded by $1$.
\end{Theorem}

\begin{Remark}
	\begin{nummer}
 	\item The choice of the matrix norm $\|\boldsymbol\cdot\|_{2\times2}$ is not crucial here. 
		It can be replaced by any other matrix norm on the space of $2\times 2$ matrices. 
	\item Our proof of the theorem relies on the bootstrap multi-scale analysis of Germinet and Klein 
		\cite{GeKl01}. In fact, the general formulation of the bootstrap multi-scale analysis in \cite{GeKl01}
		allows an immediate and straightforward application to the present setting of random block 
		operators. An alternative proof of localization has been carried out previously in \cite{ElShSo12}.
		It adapts the fractional-moment method to rather general $k\times k$-block operators for 
		$k\ge 2$ and applies in the strong-disorder regime. We would like to advertise the simplicity 
		of extending the bootstrap multi-scale analysis to our block-operator setting.
	\item Further equivalent characterizations of the region of complete localization can be found in 
		\cite{GeKl04, GeKl06}. 
	\end{nummer}
\end{Remark}

\noindent
The RAGE Theorem leads to the following well-known corollary of Theorem \ref{Anderson Local}.

\begin{Corollary}[Spectral localization]
	Under the assumptions of Theorem \ref{Anderson Local} 
	there is only pure point spectrum in $I$, that is
	\begin{equation}
		\sigma(\Hh)\cap I=\sigma_{\mathrm{pp}}(\Hh)\cap I
	\end{equation}
	holds $\PP$-a.s., and the eigenfunctions of $\Hh$ associated with eigenvalues in $I$ decay exponentially 
  at infinity.
\end{Corollary}

%
\section{Proof of the Wegner estimate}
%

The following proof of Theorem~\ref{Wegner II} is close to the one given in \cite{RandomBlockoperators}, the main difference being Lemma \ref{Lemma:....} below.

\begin{proof}[Proof of Theorem~\ref{Wegner II}]
In order to stress the dependence of the finite-volume operator on the families of random variables 
$V:=(V_n)_{n\in\Zd}$ and $B:=(B_n)_{n\in\Zd}$, we use the notation $\Hh_{L} \equiv \Hh_{L}(V,B)$ whenever appropriate. Since 
\begin{equation}
	\PP\left(\textnormal{all eigenvalues of $\Hh_L$ are non-degenerate} \right)=1,
\end{equation}
see e.g.\ \cite[Prop.\ II.1]{KuSo80}, we infer from analytic perturbation theory that for $\PP$-a.e.\
$(V,B)$ the distinct eigenvalues $E_{j} \equiv E_j(V,B)$, $j=1, \ldots, 2 |\Lambda_{L}|$, 
of $\Hh_{L}(V,B)$, which are ordered by magnitude, are all continuously differentiable (separately in each $V_{n}$ and each $B_{n}$ for $n=1,\ldots,|\Lambda_{L}|$) in the point $(V,B)$. 
For the time being we fix $E>0$ and $\epsilon>0$ with $3\epsilon<E$. Consider a switch function 
$\rho\in C^1(\mathbb R)$, i.e.\ $\rho$ is continuously differentiable, non-decreasing and obeys
$0\leq \rho\leq 1$, with $\rho(\eta)=1$ for $\eta>\epsilon$ and $\rho(\eta)=0$ for $\eta<-\epsilon$. Monotonicity gives the estimate
\begin{align}
	\label{start-int}
	\tr_{\mathcal H^2_L}\left[1_{[E-\epsilon,E+\epsilon[}\left(\Hh_{L}\right)\right]
	&\leq\sum\limits_{j=1}^{2|\Lambda_{L}|}\big[\rho\left(E_j-E+2\epsilon\right)-\rho\left(E_j-E-2\epsilon\right)\big]\nonumber\\
& = \int_{E-2\epsilon}^{E+2\epsilon}\!\d\eta\, 
\sum\limits_{j=1}^{2|\Lambda_L|}\rho'\left( E_j-\eta\right).
\end{align}
%
We infer from the chain rule that
\begin{multline}
	\label{chainrule}
	\sum_{n\in\Lambda_L} \left(\frac{\partial}{\partial V_n}+\frac{\partial}{\partial B_n}\right)
		\rho\big(E_j(V,B)-\eta\big) \\
		=  \rho'\big( E_j(V,B)-\eta\big) \sum_{n\in\Lambda_L} 
			\left(\frac{\partial}{\partial V_n}+\frac{\partial}{\partial B_n}\right)E_j(V,B)
\end{multline}
for all $j$, all $\eta$ and $\PP$-a.a.\ $(V,B)$. 
Unlike the standard Anderson model, the eigenvalues $E_j(V,B)$ are neither monotone in the $V_n$'s 
nor in the $B_n$'s, but the choice of $\epsilon$ ensures that only positive eigenvalues contribute to the $j$-sum in \eqref{start-int}. Therefore we apply Lemma \ref{Lemma:....} to \eqref{chainrule}, 
and estimate $\rho'$ in \eqref{start-int} according to
\begin{equation}
	\rho'\big( E_j(V,B)-\eta\big)\leq
	\sum_{n\in\Lambda_L} \left(\frac{\partial}{\partial V_n}+\frac{\partial}{\partial B_n}\right)
	\rho\big(E_j(V,B)-\eta\big).
\end{equation}
Taking the expectation of \eqref{start-int} and using its product structure, we obtain
\begin{align}
	\label{e-bound}
 	\mathbb E\Big\{\tr_{\mathcal H^2_L}& \big[1_{[E-\epsilon,E+\epsilon[} (\Hh_L)\big] \Big\} \\
	 \leq &\int_{E-2\epsilon}^{E+2\epsilon} \!\d\eta 
			\sum_{n\in\Lambda_L}
			\int_{\RR^{2|\Lambda_{L}|}} \! \bigg(\prod\limits_{k\in\Lambda_L}\d\mu_V(V_k) \,\d\mu_B(B_k)\bigg) 
 		\nonumber\\
	&\hspace*{0.5cm} \times \sum_{j=1}^{2|\Lambda_L|} \left(\frac{\partial}{\partial 
				V_n}+\frac{\partial}{\partial B_n}\right) \rho\big(E_j(V,B)-\eta\big).
\end{align}
Each term of the $n$-sum in the previous expression can be rewritten as 
\begin{align}
 \int_{\RR^{2|\Lambda_{L}| -2}} \! \bigg( &\prod\limits_{k\in\Lambda_L: \; k \neq n} \d\mu_V(V_k) \,\d\mu_B(B_k)\bigg) \nonumber\\
& \times \bigg[ \int_{\RR} \! \d\mu_{B}(B_{n}) \int_{\RR} \! \d\mu_{V}(V_{n}) \sum_{j=1}^{2|\Lambda_L|} 
	\frac{\partial}{\partial V_n} \rho\big(E_j(V,B)-\eta\big)  \nonumber\\
& \hspace*{.5cm} + \int_{\RR} \! \d\mu_{V}(V_{n}) \int_{\RR} \! \d\mu_{B}(B_{n}) \sum_{j=1}^{2|\Lambda_L|} 
	\frac{\partial}{\partial B_n} \rho\big(E_j(V,B)-\eta\big)\bigg].  
\end{align}
Functions like $X_{n} \mapsto F(X_{n}) := \sum_{j=1}^{2|\Lambda_L|} \rho
\big(E_j(V,B)-\eta\big)$, where $X$ stands for $V$ or $B$, are non-monotone in general. But analytic perturbation theory ensures that $F \in C^{1}(\RR)$. Moreover, $|F(x) - F(x')| \le 2$ for all $x,x'\in\RR$ by a rank-2-perturbation argument. Therefore, Lemma 5.4. in \cite{RandomBlockoperators} implies
\begin{equation}
	\label{BV}
	\int_{\RR}\! \d\mu_X(X_n) \sum\limits_{j=1}^{2|\Lambda_L|}\frac{\partial}{\partial X_n}
	\rho \big(E_j(V,B)-\eta\big)\leq 2\left\|\phi_X\right\|_{BV}
\end{equation}
for both $X=V$ and $X=B$. Thus, we conclude from \eqref{e-bound} -- \eqref{BV} that
\begin{equation}
	\label{wegner-loc}
	\mathbb E\Big\{\tr_{\mathcal H^2_L} \big[1_{[E-\epsilon,E+\epsilon[} (\Hh_L)\big] \Big\} 
	\leq 8\epsilon \left|\Lambda_L\right| 
	\big( \left\|\phi_V\right\|_{BV} + \left\|\phi_B\right\|_{BV} \big)
\end{equation}
for every $E>0$ and every $0< \epsilon < E/3$.
This bound and dominated convergence establish Lipschitz continuity of the integrated density of states
$\mathds N$ on $\mathbb R_{> 0}$ with Lipschitz constant $2 \left(\left\|\phi_V\right\|_{BV}+\left\|\phi_B\right\|_{BV}\right)$.  But due to the symmetry of the spectrum, see Remark~\ref{symm-spec}, 
this extends to $\RR \setminus \{0\}$. Furthermore, since $\mathds N$ is a continuous function on the whole real line $\RR$ -- which follows from standard arguments as in  \cite[Thm.\ 5.14]{Kirsch:Invitation} -- this yields Lipschitz continuity on $\mathbb R$ with the same constant. 
\end{proof}

\noindent
One of the main estimates in the previous proof is provided by the following deterministic result. 

\begin{Lemma}\label{Lemma:....}
	Assume that $H_L\geq0$, $B_L\geq 0$ and let $E(V,B) >0$ be a simple eigenvalue of 
	$\Hh_L(V,B)$. Then we have 
	\begin{equation}
		\sum_{n\in\Lambda_L}\left(\frac{\partial}{\partial V_n} + 
			\frac{\partial}{\partial B_n}\right)E(V,B)\geq 1.
	\end{equation}
\end{Lemma}

\begin{proof}
Let $\Psi = (\psi_1,\psi_2)$ be a normalized eigenvector corresponding to the eigenvalue 
$E \equiv E(V,B)$ of the operator $\Hh_{L} \equiv \Hh_L(V,B)$, i.e.\ 
$\langle \psi_1,\psi_1\rangle+\langle \psi_2,\psi_2\rangle=1$ and
\begin{equation}
 	\label{ev-eq}
	\begin{split}
		H_L\psi_1+B\psi_2&=E\psi_1,\\
		B\psi_1-H_L\psi_2&=E\psi_2.
	\end{split}
\end{equation}
The Feynman-Hellmann formula for a non-degenerate eigenvalue and \eqref{ev-eq} imply
\begin{align}
 	E\sum\limits_{j\in\Lambda_L} &\left(\frac{\partial}{\partial V_j} + 
 	\frac{\partial}{\partial B_j}\right) E \nonumber\\
	&=E \,\big( \langle \psi_1, \psi_1\rangle - \langle \psi_2, \psi_2\rangle 
		+ \langle \psi_1, \psi_2\rangle + \langle \psi_2,\psi_1\rangle \big)\nonumber\\
	&= \left\langle \psi_1, H_L\psi_1 +B\psi_2\right\rangle - 
			\left\langle B\psi_1-H_L\psi_2,\psi_2\right\rangle
 		  + \left\langle \psi_1,B\psi_1-H_L\psi_2\right\rangle \nonumber\\ 
 	& \phantom{==} + \left\langle \psi_2,H_L\psi_1+B\psi_2\right\rangle \nonumber\\
&=\left\langle \psi_1,H_L\psi_1\right\rangle+\left\langle \psi_2,H_L\psi_2\right\rangle+\left\langle 
\psi_1,B\psi_1\right\rangle+\left\langle \psi_2,B\psi_2\right\rangle\label{FeHe}.
\end{align}
In the last step we used that the operator $\Hh_{L}$ is a real symmetric matrix and, therefore, the eigenvector $\Psi$ can be chosen to be real.
Since $B\geq 0$, we have
\begin{equation}
	\left\langle \psi_1,B\psi_1\right\rangle + \left\langle \psi_2,B\psi_2\right\rangle 
	\ge 	\left\langle \psi_1,B\psi_2\right\rangle + \left\langle \psi_2,B\psi_1\right\rangle.
\end{equation}
This and $H_{L} \ge 0$ yield the lower bound
\begin{equation}
	\langle \psi_1, H_L\psi_1+B\psi_2\rangle + \langle \psi_2,B\psi_1-H_L\psi_2\rangle
	=E \,\big( \langle \psi_1,\psi_1\rangle + \langle \psi_2,\psi_2\rangle\big) =E
\end{equation}
for the r.h.s.\ of \eqref{FeHe}.
\end{proof}

%
\section{Proof of Lifschitz tails}
\label{sec:lif}
%

In this section we prove Theorems~\ref{LifschitzII} and~\ref{Lifschitz-lower}.

\noindent
The idea behind the proof of Theorem~\ref{LifschitzII} is to estimate the integrated density of states of $\Hh$ in terms of 
the integrated density of states of the operator
\begin{equation}
\Hh(\beta):=\begin{pmatrix} H&\beta\one\\ \beta\one&-H\end{pmatrix},
\end{equation}
on $\mathcal{H}^{2}$, where $\beta$ is as in Theorem~\ref{LifschitzII} and $\one$ denotes the unit 
operator on $\ell^{2}(\Zd)$.
This is useful because  we explicitly know the relation between the spectra of $\Hh(\beta)$ and $H$, and because the discrete Schr\"odinger operator $H$ of the Anderson model 
exhibits Lifschitz tails at the edges of its spectrum. For the lower spectral edge of $H$ the upper Lifschitz-tail estimate is summarized in the next lemma, for a proof see e.g. \cite{Carmona, PasturFigotin, springerlink:10.1007/3-540-51783-9_23}.

\begin{Lemma}[Upper Lifschitz-tail estimate for $\boldsymbol H$]
	\label{StandLif}
	Let $H$ be the discrete random Schr\"o\-dinger operator of the Anderson model as in 
	Definition~\ref{Def:Hh}. Assume in addition that the single-site probability measure $\mu_{V}$ 
	is not concentrated in a single point.
	 Then, the integrated density of states $N_H$ of the operator $H$ obeys
	\begin{equation}
		\limsup_{\epsilon\searrow 0} \frac{\ln\left|\ln\left[N_H(\lambda + \epsilon)\right]
				\right|}{\ln\epsilon}\leq-\frac{d}{2},
	\end{equation}
	where $\lambda :=  \inf\supp \mu_{V}= \inf \sigma(H)$ is the infimum of the almost-sure spectrum of $H$.
\end{Lemma}

\noindent
The remaining arguments needed for the proof of Theorem~\ref{LifschitzII} are all deterministic.
The next lemma, which is a particular case of \cite[Thm.\ 1.9.1]{Tretter:Spectraltheory}, 
provides a variational principle for the positive spectrum 
of the finite-volume block operator $\Hh_{L}$. 

\begin{Lemma}[Min-max-max principle]
	\label{min-max-max}
        Given $A,B$ and $D$ self adjoint operators on $\mathcal H=l^2(\Lambda_L)$ with $A>-D$, define the block operator $\mathds A:=\begin{pmatrix}
                                      A & B\\
                                      B & -D
                                     \end{pmatrix}$
        on $\mathcal H^2$.
 	Then
	\begin{nummer}
         \item  there are precisely $|\Lambda_L|$ eigenvalues of $\mathds A$, $\lambda_1,...,\lambda_{|\Lambda_L|}$, with $\lambda_j>\sup\sigma(-D)$ and
	 \item  the eigenvalues $\lambda_j>\sup\sigma(-D)$, $j=1, \ldots, |\Lambda_{L}|$, ordered by magnitude and repeated according to their multiplicity, are given by
                 \begin{multline}
		        \lambda_{j} = \min_{\substack{\mathcal V\subset \ell^{2}(\Lambda_{L}): \\ \dim \mathcal V=j}}\ 
		        \max_{\substack{f\in\mathcal V:\\ \left\|f\right\|=1}}\
		        \max_{\substack{g\in \ell^{2}(\Lambda_{L}): \\ \left\|g\right\|=1}}
		        \Bigg\{ \frac{\langle f,Af \rangle - \langle g,Dg\rangle}{2} \\
		      + \sqrt{ \left( \frac{\langle f,Af \rangle + \langle g,Dg \rangle}{2}\right)^2 
	              + \left|\left\langle f,Bg\right\rangle \right|^2} \; \Bigg\}.
	         \end{multline}
	%
        \end{nummer}
\end{Lemma}

\noindent
This variational characterization will serve to relate the positive spectrum of $\Hh_{L}$ 
to that of $\Hh_L(\beta)$, which is the restriction of $\Hh(\beta)$ to $\mathcal{H}^{2}_{L}$ 
in analogy with Definition~\ref{Def:HhL}.
Finally, we relate the spectrum of $\Hh_{L}(\beta)$ to that 
of its diagonal block $H_{L}$.

\begin{Lemma}[\protect{\cite[Prop.\ 3.1]{RandomBlockoperators}}]
	\label{Lemma:H beta}
	The spectrum of $\Hh_{L}(\beta)$ is given by 
	\begin{equation}
		\sigma\big(\Hh_{L}(\beta)\big) = \big\{\pm\sqrt{E^2+\beta^2}:\ E\in\sigma(H_{L})\big\}.
	\end{equation}
\end{Lemma}
 
\medskip
\noindent
Now we are prepared for the 

\begin{proof}[Proof of Theorem~\ref{LifschitzII}]

	Since $H \geq 0$ we have $H_L>0$ and can apply Lemma~\ref{min-max-max}. Setting $f=g$ there and noting 
	that $\PP$-a.s.\ $\beta = \inf \sigma(|B|)$, 	we infer
	\begin{align}
		\lambda_j & \ge 	\min_{\substack{\mathcal V\subset \ell^{2}(\Lambda_{L}):\\ \dim \mathcal V=j}}\ 
			\max_{\substack{f\in\mathcal V:\\ \left\|f\right\|=1}}
			\sqrt{ \langle f, H_Lf \rangle^2 + \langle f,Bf\rangle ^2} \nonumber \\
		& \ge 	\min_{\substack{\mathcal V\subset \ell^{2}(\Lambda_{L}):\\ \dim \mathcal V=j}}\ 
			\max_{\substack{f\in\mathcal V:\\ \left\|f\right\|=1}}
			\sqrt{ \langle f, H_Lf \rangle^2 + \beta^{2}} \nonumber \\
		& = \left[ \Bigg( \min_{\substack{\mathcal V\subset\ell^{2}(\Lambda_{L}): \\ \dim \mathcal V=j}}\ 
			\max_{\substack{f\in\mathcal V:\\ \left\|f\right\|=1}} \langle f, H_L f \rangle \Bigg)^2+\beta^2 
			 \right]^{1/2} 
	\end{align}
	for every $j=1, \ldots, |\Lambda_{L}|$. We denote the positive eigenvalues of $\Hh_{L}(\beta)$
	by $0 < \mu_{1} \le \ldots \le \mu_{|\Lambda_{L}|}$.	The min-max principle for $H_{L}$ and 
	Lemma~\ref{Lemma:H beta} then imply
	\begin{equation}
		\label{Lifschitz:Label}
 		\lambda_j \ge \mu_{j}
	\end{equation}
	for every $j=1, \ldots, |\Lambda_{L}|$.
        \noindent
	Symmetry of the spectra of $\Hh_L$ and $\Hh_{L}(\beta)$, see Remark~\ref{symm-spec}, 
	the strict positivity $H_{L} > \inf \sigma(H) = \lambda \ge 0$ and Lemma~\ref{doublegap} imply
	\begin{equation}
		\label{half-half}
		\mathds N_{\Hh_L} \big(\sqrt{\lambda^2+\beta^2}\big) 
		= 
		\mathds N_{\Hh_{L}(\beta)} \big(\sqrt{\lambda^2+\beta^2}\big) 
		=\frac{1}{2}.
	\end{equation} 
	Setting $E:=\sqrt{\lambda^2+\beta^2}+\epsilon$ for $\epsilon >0$, Eqs.\ \eqref{half-half} and 
	\eqref{Lifschitz:Label} give the estimate
	\begin{align}
		\mathds N_{\Hh_{L}}(E) - \mathds N_{\Hh_{L}} &\big(\sqrt{\lambda^2 +\beta^2}\big) \nonumber\\
	  &\leq \mathds N_{\Hh_{L}(\beta)}(E) - \mathds N_{\Hh_{L}(\beta)} 
			\big(\sqrt{\lambda^2+\beta^2}\big) \nonumber\\
		& = \frac{1}{2|\Lambda_{L}|}\, \big| \big\{ \,\mu\in\sigma\big(\Hh_{L}(\beta)\big):\ 
			\mu\in \big[ \sqrt{\lambda^2+\beta^2}, E \big[ \;\big\}\big|\nonumber\\
		& = \frac{1}{2|\Lambda_{L}|}\, \big| \big\{ \widetilde\mu\in\sigma(H_{L}):\ \widetilde{\mu}\in \big[\lambda,
			\sqrt{E^2-\beta^2}\big[ \;\big\}\big| \nonumber\\ 
		&= \frac{1}{2}\, N_{H_{L}} \big(\sqrt{E^2-\beta^2}\big),
	\end{align}
	where we have used Lemma~\ref{Lemma:H beta} for the second equality.
	Therefore we get in the limit $L\to\infty$ and using Lemma~\ref{StandLif}
	\begin{align}
		& \limsup_{\epsilon\searrow 0} \frac{ \ln\big|\ln\big[\mathds N(\sqrt{\lambda^2+\beta^2}
			+ \epsilon) - \mathds N(\sqrt{\lambda^2+\beta^2})\big]\big|}{\ln\epsilon} 
			\nonumber\\
		& \hspace*{3cm}\le \limsup_{\epsilon\searrow0}\frac{\ln\Big| 
			\ln N_H\Big( \big[ (\sqrt{\lambda^2+\beta^2}+\epsilon)^2-\beta^2 \big]^{1/2}\Big)\Big|}{\ln\epsilon}\nonumber\\
		& \hspace*{3cm} = \limsup_{\widetilde{\epsilon}\searrow 0} \frac{\ln\left| 
			\ln N_H( \lambda + \widetilde{\epsilon})\right|}{\xi\ln\widetilde{\epsilon}} \nonumber\\ 
		& \hspace*{3cm}	
			\leq -\frac{d}{2\,\xi}
	\end{align}
	with $\xi = 1$ in all cases except the case of 
	$\lambda=0$ and $\beta\neq 0$, where $\xi=2$.
\end{proof}

\noindent
In the remaining part of this section we turn to the lower bound for Lifschitz tails.

\begin{proof}[Proof of Theorem~\ref{Lifschitz-lower}]

  We use Dirichlet-Neumann bracketing. Therefore  we define, 
  following \cite[Def.~4.6]{RandomBlockoperators}, the 
  \emph{Dirichlet-bracketing} restriction of the block operator as
  \begin{equation}
  	\Hh_L^+:=\begin{pmatrix} H_L^D & B\\B &-H_L^N \end{pmatrix} ,
  \end{equation}

  \noindent
  where $H_L^N$ and $H_L^D$ denote the restriction of $H$ to the cube $\Lambda_L$, $L\in\NN$, 
  with Neumann, respectively Dirichlet boundary conditions on the Laplacian; 
  for a precise definition see \cite[Sect.\ 5.2]{Kirsch:Invitation}.
  Setting $\widetilde H^D_L:=H^D_L-\lambda\one$, $\widetilde B:=B-\beta\one$ and 
  using Lemma~\ref{min-max-max}, we obtain for the $j$-th positive eigenvalue, $j=1,\ldots,|\Lambda_{L}|$,
        \begin{multline}
		 \lambda_j (\Hh_L^D) = \min_{\substack{\mathcal V\subset \ell^{2}(\Lambda_{L}): \\ \dim \mathcal V=j}}\ 
			\max_{\substack{f\in\mathcal V:\\ \left\|f\right\|=1}}\
			\max_{\substack{g\in \ell^{2}(\Lambda_{L}): \\ \left\|g\right\|=1}}
			 \Bigg\{ \frac{\langle f,\widetilde H^D_Lf \rangle - \langle g,\widetilde H^N_Lg\rangle}{2} \\
		+ \sqrt{ \bigg( \lambda+ \frac{\langle f,\widetilde H^D_Lf \rangle + \langle g,\widetilde H^N_Lg \rangle}{2}\bigg)^2 
				+ \big|\beta+ \langle f,\widetilde Bg\rangle\big|^2} \; \Bigg\}.\label{LifschitzIII}
	\end{multline}
  The elementary inequality 
  \begin{equation}
  	\sqrt{(\lambda+a)^2+(\beta+b)^2}\leq \sqrt{\lambda^2+\beta^2}+a+b
  \end{equation}
	holds for every $a,b \ge 0$ and $\lambda,\beta\in \RR$. 

  \noindent
  Together with the estimate $\big|\beta+ \langle f,\widetilde Bg\rangle\big| \le 
   |\beta| + \langle f,\widetilde B^{2} f\rangle^{1/2}$,
  this yields
   \begin{align}
    \lambda_j (\Hh_L^D) \le \sqrt{\lambda^2+\beta^2}   
    +  \min_{\substack{\mathcal V\subset \ell^{2}(\Lambda_{L}): \\ \dim \mathcal V=j}}\ 
			\max_{\substack{f\in\mathcal V:\\ \left\|f\right\|=1}}\ \Big\{
			 \langle f,\widetilde H_L^D f\rangle + \langle f,\widetilde B^2 f\rangle^{1/2} \Big\}.
   \end{align}

   \noindent
	On the other hand, \eqref{LifschitzIII} implies
	\begin{equation}
		\label{ev-lb}
	 	\lambda_j (\Hh_L^D) > \sqrt{\lambda^2+\beta^2} 
	\end{equation}
	for every $j=1,\ldots,|\Lambda_{L}|$. From this and Lemma ~\ref{min-max-max} we conclude that 
	$\EE \big[\mathds{N}_{\Hh_{L}^{+}}(\sqrt{\lambda^{2} + \beta^{2}})\big] =1/2$. 
        Similarly, using the symmetry of the spectrum and continuity of the integrated density of states (cf.\ the proof of \cite[Lemma 5.13]{Kirsch:Invitation}), we obtain $\EE \big[\mathds{N}(\sqrt{\lambda^{2} + \beta^{2}})\big] =1/2$.
  These two equalities and the estimate $\mathds{N}(E) \ge \EE 
  \big[\mathds{N}_{\Hh_{L}^{+}}(E)\big]$ for every 
	$E\in\RR$ \cite[Lemma 4.8(ii)]{RandomBlockoperators} yield
  \begin{align}
   	\mathds N(\sqrt {\lambda^2+\beta^2} &+\epsilon) - 
    \mathds N(\sqrt{\lambda^2+\beta^2}) \nn\\
    & \ge \EE \big[\mathds{N}_{\Hh_{L}^{+}}(\sqrt{\lambda^{2} + \beta^{2}} +\epsilon)\big] 
      - \EE \big[\mathds{N}_{\Hh_{L}^{+}}(\sqrt{\lambda^{2} + \beta^{2}})\big]  \nn\\
    & \ge \frac 1 {2|\Lambda_L|} \, \PP\left( \lambda_{1}  (\Hh_L^D) \in
    [\sqrt{\lambda^2+\beta^2},\sqrt{\lambda^2+\beta^2}+\epsilon[  \right) \nn\\
    &\ge \frac 1 {2|\Lambda_L|} \,\PP\left(\langle \psi,\widetilde H_L^D \psi\rangle + 
    	\langle \psi,\widetilde B^2 \psi\rangle^{ 1/ 2}< \epsilon \right) \label{Lifschitz4}
  \end{align}
 	for every $L\in \NN$, $\epsilon>0$ and every normalized test function $\psi \in l^2(\Lambda_L)$.
        
        \noindent
	Following \cite[Sect.\ 6.3]{Kirsch:Invitation}, we choose
   $\psi :=\frac 1 {\|\psi_1\|}\psi_1(n)$, where 
   $\psi_1(n) :=\frac L 2 - |n|_\infty$ for $n\in\Lambda_{L}$. This implies
   $\langle \psi, H_{0,L}^D \psi\rangle\le c_0 L^{-2}$ with some constant $c_{0}>0$.
   Next we choose $L$ to be the smallest integer such that
   \begin{equation}
      c_0L^{-2}<\epsilon/ 2
   \end{equation}
   and estimate
   \begin{align}
    \PP\Big(\langle \psi,\widetilde H_L^D \psi\rangle &+ 
    	\langle \psi,\widetilde B^2 \psi\rangle^{ 1/ 2}< \epsilon \Big) \nn\\
		&\ge  \PP\left(\langle \psi, (V -\lambda\one) \psi\rangle + 
    	\langle \psi, (B -\beta\one)^2 \psi\rangle^{ 1/ 2}< \epsilon/2 \right) \nn\\
    &\ge \PP \Big( \forall\ n \in\Lambda_{L}:\ V(n) -\lambda <\epsilon/4 \;\;\text{and~} \;
   		|B(n) -\beta| <\epsilon/4\Big)   \nn\\
    &=  \big\{\mu_{V}\big( [\lambda, \lambda +\epsilon/4[ \big)\big\}^{|\Lambda_{L}|} \;
        \big\{\mu_{B}\big( [\beta - \epsilon/4, \beta +\epsilon/4[ \big)\big\}^{|\Lambda_{L}|} .
   \end{align}
 
 \noindent
 The theorem now follows with \eqref{Lifschitz4} and the assumption \eqref{Lifschitz6}.
\end{proof}

%
\section{Proof of localization}
%

Our proof relies on the bootstrap multi-scale analysis introduced in \cite{GeKl01}, which 
yields complete localization in a rather general setting. 
Apart from one natural adaptation for multiplication operators -- see below --  we are only left to check whether the assumptions on the random operator are fulfilled by our model. We start with some notions.

\begin{Def}
We introduce the \emph{boundary} of a cube $\Lambda \subset\Zd$ by 
\begin{equation}
	\partial\Lambda := \big\{ (n,m)\in\Zd \times \Zd:\  |n-m| =1, \; n\in\Lambda,\ m\notin\Lambda 
	\text{\; or\; }  n\notin\Lambda,\, m\in\Lambda \big\},
\end{equation}
its \emph{inner boundary} by 
\begin{equation}
 	\partial^i\Lambda := \big\{ n \in \Lambda : \exists\; m \not\in \Lambda 
			\text{ such that } |n - m| =1  \big\}
\end{equation}
and its \emph{outer boundary} by
\begin{equation}
 			\partial^o\Lambda := \big\{ n \not\in \Lambda : \exists\; m \in \Lambda 
			\text{ such that } |n - m| =1  \big\}.
\end{equation}
We write $\Lambda_1\sqsubset\Lambda_2$ if $\partial \Lambda_1\subset\Lambda_2\times\Lambda_2$.
Furthermore for $\Lambda_1\sqsubset\Lambda_2\subseteq \mathbb Z^d$ we define the \emph{boundary operator} 
$\Gamma_{\Lambda_1}^{\Lambda_2}\equiv\Gamma_{\Lambda_1}$ on $\ell^{2}(\Lambda_{2})$ in terms of its matrix elements
\begin{equation}
	\langle\delta_{n},\Gamma_{\Lambda_1}\delta_{m}\rangle := \ccases{-1,  & (n,m)\in\partial\Lambda_1, \\
		0, & (n,m)\in (\Lambda_{2} \times \Lambda_{2}) \setminus \partial\Lambda_1.
	}
\end{equation}
We lift $\Gamma_{\Lambda_{1}}$ to a bounded operator on 
$\ell^{2}(\Lambda_{2}) \oplus \ell^{2}(\Lambda_{2})$ by setting 

\begin{equation}
 	\bbGamma_{\Lambda_{1}} := \Gamma_{\Lambda_{1}} \oplus (-\Gamma_{\Lambda_{1}}).
\end{equation}
In contrast, given subsets $\Lambda \subset \Lambda'\subseteq\Zd$, we lift the multiplication operator $1_{\Lambda}$ on $\ell^{2}(\Lambda')$, corresponding to the indicator function of $\Lambda$, to the sum space $\ell^{2}(\Lambda') \oplus \ell^{2}(\Lambda')$ by setting 
\begin{equation}
 	\11_{\Lambda} := 1_{\Lambda} \oplus 1_{\Lambda}.
\end{equation}
In slight abuse of notation we also write $\mathds 1_n:=\mathds 1_{\{n\}}$ for $n\in\Zd$. 
Finally, given an energy $E \not\in\sigma(\Hh_{\Lambda})$, we use the abbreviation $\mathds G_\Lambda(E):=\left(\Hh_\Lambda-E\right)^{-1}$ for the resolvent of $\Hh_{\Lambda}$.
\end{Def}

\begin{proof}[Proof of Theorem~\ref{Anderson Local}]
 	We apply \cite[Thm.\ 3.8]{GeKl01} on the Hilbert space $\mathcal{H}^{2}$, with the random operator 
	$\Hh$ and with $\11_{\Lambda}$ playing the role of the multiplication operator 
	$\chi_\Lambda$ in \cite{GeKl01}.
	The deterministic Assumptions \emph{SLI} and \emph{EDI} will be checked in Lemmas~\ref{sli} 
	and Lemma~\ref{edi} below. 
	We note a slight structural difference between the statement of Lemma~\ref{edi} and the 
	\emph{EDI}-property in \cite{GeKl01}: the factor $\|\mathds 1_{\partial^o\Lambda}\Psi\|$ 
	in \eqref{eq:edi} evaluates $\Psi$ outside the cube $\Lambda$. However, this factor plays 
	only a role in the proof of Lemma 4.1 in \cite{GeKl01}, and Eq.\ (4.3) -- (4.4) in that proof 
	show that this difference is irrelevant. 

\noindent       
	The next important hypothesis of Thm.~3.8 in \cite{GeKl01} is the Wegner Assumption \emph{W}, which follows from Theorem~\ref{Wegner I} or~\ref{Wegner II} for our model with $b=1$ (more precisely from the finite-volume estimates -- e.g.\ \eqref{wegner-loc} -- in the proofs of those theorems). 
The remaining assumptions \emph{IAD}, \emph{NE} and \emph{SGEE} are obviously correct because we work with 
a discrete model with i.i.d.\ random coupling constants.  
Finally, the initial-scale estimate follows from  Theorem~\ref{Initial-scale} below, 
see also Remark 3.7 in \cite{GeKl01}. 

\noindent
Having collected all the aforementioned properties, Cor.\ 3.12 of \cite{GeKl01} implies that the claim of Thm.~3.8 holds for all energies in some interval $I:=[-a,a]$, where $a > \sqrt{\lambda^{2}+ \beta^{2}}$ 
so that $I$ overlaps with the almost-sure spectrum of $\Hh$ according to Lemma~\ref{Lemma:Spectrum1}. The claim of Thm.~3.8 then reads 
\begin{equation}
	\label{gk-statement}
 	\EE \left( \sup_{\|f\|_{\infty} \le 1}  \big\| \11_{n} 1_I(\Hh) f(\Hh) \11_{m} \big\|_{HS}^{2} \right) 
	\le C_{\zeta} \e^{-|n-m|^\zeta}
\end{equation}
for all $n,m \in\Zd$. Here, $\|\mathds{A}\|_{HS}$ is the Hilbert-Schmidt norm of an operator $\mathds{A}$ on $\mathcal{H}^{2}$. To get to our formulation in \eqref{eq:dyn-loc} we remark that 
\begin{equation}
	\|\11_n \mathds A \11_m\|_{HS} = \|\mathds A(n,m)\|_{2\times2},
\end{equation}
where, on the right-hand side, we use the notation introduced in \eqref{matrix-element}, 
and $\|\boldsymbol\cdot\|_{2\times 2}$ stands for the Hilbert-Schmidt norm of a $2\times2$-matrix. 
Replacing the latter by any other norm on the $2\times2$-matrices as in \eqref{eq:dyn-loc},
merely requires a possible adjustment of the constant $C_{\zeta}$.
\end{proof}

\noindent
Next we deal with the deterministic assumptions required by the bootstrap multi-scale analysis. 
The first one is a consequence of the geometric resolvent equation \eqref{gre}. 

\begin{Lemma}[\emph{SLI}]
\label{sli}
Let $\Lambda_1\sqsubset \Lambda_2 \sqsubset \Lambda_3$. Then we have for $E\notin(\sigma(\Hh_{\Lambda_2})\cup\sigma(\Hh_{\Lambda_3}))$ the inequality
\begin{equation}
\left\| \mathds 1_{\partial^i\Lambda_3} \mathds G_{\Lambda_3}(E) \mathds 1_{\Lambda_1} \right\|\leq
 \gamma \left\| \mathds 1_{\partial^i\Lambda_3}\mathds G_{\Lambda_3}(E)\mathds 1_{\partial^o \Lambda_2}\right\|    \left\| \mathds 1_{\partial^i\Lambda_2}\mathds G_{\Lambda_2}(E) \mathds 1_{\Lambda_1} \right\|,
\end{equation}
where $\gamma>0$ depends only on the space dimension $d$ and the norm is the operator norm.
\end{Lemma}

\begin{proof}
The identity
\begin{equation}
	\label{decomp}
	\Hh_{\Lambda_3} = (\Hh_{\Lambda_2} \oplus \Hh_{\Lambda_3\backslash \Lambda_2}) + 
	\bbGamma _{\Lambda_2}
\end{equation}
and the resolvent equation imply
\begin{align}
	\label{gre}
	\mathds 1_{\partial^i\Lambda_3} \mathds G_{\Lambda_3}(E)  \mathds 1_{\Lambda_1}
	&= -\mathds 1_{\partial^i \Lambda_3} 
        \mathds G_{\Lambda_3}(E) \,\bbGamma_{\Lambda_2} \mathds G_{\Lambda_2}(E)  \mathds 1_{\Lambda_1} \\
  &= -\mathds 1_{\partial^i \Lambda_3} 
        \mathds G_{\Lambda_3}(E) \mathds 1_{\partial^o \Lambda_2}  \bbGamma_{\Lambda_2} \mathds 1_{\partial^i \Lambda_2}  \mathds G_{\Lambda_2}(E) \mathds 1_{\Lambda_1}, \nonumber
\end{align}
where we used that $ \bbGamma_{\Lambda_2} \mathds 1_{\Lambda_2} =
\mathds 1_{\partial^o \Lambda_2}  \bbGamma_{\Lambda_2} \mathds 1_{\partial^i \Lambda_2}$.
Taking the norm and observing that $\gamma := \|\bbGamma_{\Lambda_2} \|$ depends only on the space dimension $d$, yields the statement. 
\end{proof}

\noindent
A similar argument proves

\begin{Lemma}[\emph{EDI}]
\label{edi}
Let $\Psi$ be a generalized eigenfunction of $\Hh$ with generalized eigenvalue $E$ and let $\gamma$ be the constant from the previous lemma.
Then we have for any $\Lambda$ such that $E\not\in \sigma(\Hh_\Lambda)$ and $n\in\Lambda$
\begin{equation}
\label{eq:edi}
\left\| \mathds 1_{n}\Psi \right\| \leq \gamma \left\| \mathds 1_{n} \mathds G_\Lambda(E) \mathds 1_{\partial^i\Lambda}\right\|  \left\|\mathds 1_{\partial^o\Lambda}\Psi\right\|.
\end{equation}
\end{Lemma}

\begin{proof}
We infer from \eqref{decomp} with $\Lambda_{3} = \Zd$ and $\Lambda_{2} = \Lambda$ that
\begin{equation}
 	(\Hh_{\Lambda} \oplus \Hh_{\Zd\setminus\Lambda}- E) \Psi = - \bbGamma_{\Lambda}  \Psi. 
\end{equation}
Since $E\not\in \sigma(\Hh_\Lambda)$ and $n\in\Lambda$, this implies 
$\11_{n}\Psi = -\11_{n}\mathds{G}_{\Lambda}(E) \bbGamma_{\Lambda}\Psi $. The identity 
$\11_{\Lambda}\bbGamma_{\Lambda} = \11_{\partial^{i}\Lambda} \bbGamma_{\Lambda} \11_{\partial^{o}\Lambda}$ 
and taking norms finishes the proof.
\end{proof}

\noindent
The remaining part of this section is concerned with the verification of the initial-scale estimate.

\begin{Def}
 	Let $\theta >0$ and $E\in \RR$. A cube $\Lambda_{L} \subset\Zd$, $L\in 6\NN$, is 
	\emph{$(\theta,E)$-suitable}, if $E \not\in \sigma(\Hh_{L})$ and
	\begin{equation}
 		\| \11_{\partial^{i}\!\Lambda_{L}} \mathds{G}_{\Lambda_{L}}(E) \11_{\Lambda_{L/3}} \| < L^{-\theta}.
	\end{equation}
\end{Def}

\begin{Theorem}[Initial estimate]
	\label{Initial-scale}
	Consider the random block operator $\Hh$ of Definition~\ref{Def:Hh} and assume the hypotheses of 
	Theorem~\ref{LifschitzII}. Then there exist constants $\theta >d$ and $p>0$ such that for 
	every length $L \in 6\NN$ 	sufficiently large the following holds: there 
	exists an energy $a_{L} > \sqrt{\lambda^{2} + \beta^{2}}$ such that 
	\begin{equation}
		\label{suit-prob}
		\PP \big( \Lambda_{L} \text{ is $(\theta,E)$-suitable}\big) > 1- L^{-p}
	\end{equation}
	for every energy $E \in [-a_{L}, a_{L}]$. 
\end{Theorem}

\noindent
We use Lifschitz tails at the internal bad edges to prove Theorem~\ref{Initial-scale}. 
Lifschitz tails arise from a small probability for finding an eigenvalue close to the spectral edge. This mechanism also yields the high probability for the event in \eqref{suit-prob}.
As in the proof of Lifschitz tails for $\Hh$ in Sect.~\ref{sec:lif}, we will reduce this to a corresponding statement for $H$.

\begin{Lemma}[Lifschitz-tail estimate \protect{\cite[Eq.\ (11.23)]{Kirsch:Invitation}}]
	\label{Lemma:Initial}
	Let $H$ be the discrete random Schr\"o\-dinger operator of the Anderson model as in 
	Definition~\ref{Def:Hh}. Assume in addition that the single-site probability measure $\mu_{V}$ 
	is not concentrated in a single point and let $\lambda :=  \inf\supp \mu_{V}$ 
	be the infimum of the almost-sure spectrum of $H$. Then, given any $C,p>0$, we have for every 
	$L\in\NN$ sufficiently large
	\begin{equation}
		\PP\big( \inf \sigma(H_{L}) \leq \lambda + C L^{-1/2} \big)\leq \frac{1}{L^p}.
	\end{equation}
\end{Lemma}

\noindent
As a second ingredient for the initial-scale estimate we need 
some natural decay of the Green function of $\Hh_{L}$. 

\begin{Lemma}[Combes-Thomas estimate]
	For $L\in\NN$ consider the finite-volume block operator $\Hh_{L}$ of Definition~\ref{Def:HhL}. 
	Fix $E\in\RR$ with 
	$\dist(E,\sigma(\Hh_{L}))\geq\delta$ for some $\delta\leq 1$. Then we have for all $n,m\in\mathbb Z^d$
	that
	\begin{equation}
		\| \11_{n} \mathds{G}_{\Lambda_{L}}(E) \11_{m} \| \leq\frac{4}{\delta}\, 
		\e^{- (\delta /12d) \, |n-m|}.
	\end{equation}
\end{Lemma}

\begin{proof}
	We have patterned the lemma after \cite[Thm.\ 11.2]{Kirsch:Invitation}, and its proof 
	follows from a straightforward adaptation to random block operators of the proof there. 
	Details can be found in \cite{Diplomarbeit}. 
\end{proof}

\noindent
We are now ready for the 

\begin{proof}[Proof of Theorem~\ref{Initial-scale}]
Fix $\theta >d$, let $L\in6\NN$, set $a_{L} := \sqrt{\lambda^{2}+\beta^{2}} + L^{-1/2}$ and pick any $E \in [-a_{L},a_{L}]$. Assuming that the event
\begin{equation}
	\label{event}
 	\inf\sigma(|\Hh_{L}|) > a_{L} + L^{-1/2} 
\end{equation}
holds, then the Combes-Thomas estimate yields
\begin{equation}
	\left\|\11_{n} \mathds{G}_{\Lambda_{L}}(E) \11_{m} \right\| 
	\leq 4 \sqrt{L} \e^{- |n-m| /  (12d \sqrt{L})}
	\leq 4 \sqrt{L} \e^{- \sqrt{L} /(48d)}
\end{equation} 
for all $n\in\partial^{i}\!\Lambda_{L}$ and all $m\in \Lambda_{L/3}$.
Thus, provided $L$ is sufficiently large, the event \eqref{event} implies that
the cube $\Lambda_{L}$ is $(\theta,E)$-suitable.
Negating this implication, we conclude 
\begin{equation}
	\PP \big( \Lambda_{L} \text{ is not $(\theta,E)$-suitable}\big) 
	\leq \PP \big( \inf\sigma(|\Hh_{L}|) \le a_{L} + L^{-1/2} \big).
\end{equation}
The symmetry of the spectrum and the ordering (\ref{Lifschitz:Label}) of the eigenvalues of the operators $\Hh_{L}$ and $\Hh_{L}(\beta)$ gives 
\begin{align}
	\PP \big( \Lambda_{L} \text{ is not $(\theta,E)$-suitable}\big)
	&\leq  \PP \big( \inf\sigma(|\Hh_{L}(\beta)|) \le a_{L} + L^{-1/2} \big) \nonumber\\
  &\leq  \PP \Big( \inf\sigma(|\Hh_{L}(\beta)|) \le \sqrt{(\lambda+CL^{-1/2})^2+\beta^2}\Big) \nonumber\\
  &=  \PP \big( \inf\sigma(H_{L}) \le \lambda+CL^{-1/2} \big),
\end{align}
where $C\ge1$ is some $L$-independent constant, and the equality in the last line relies on Lemma~\ref{Lemma:H beta}. The claim now follows from Lemma~\ref{Lemma:Initial}.
\end{proof}

\section*{\normalsize\large{Acknowledgment}}
	We thank Werner Kirsch for stimulating discussions and the German Research Council (Dfg) for partial financial support through Sfb/Tr 12. 


\end{document}